\documentclass[energies,article,submit,moreauthors,dvipdfm,12pt,a4paper]{nolineno_mdpi}
\lastpage{x}
\doinum{10.3390/------}
\pubvolume{xx}
\pubyear{2015}
\history{Received: xx / Accepted: xx / Published: xx}

\usepackage{epsfig} 
\usepackage{mathptmx} 
\usepackage{times} 
\usepackage{amsmath} 
\usepackage{amssymb}  
\usepackage{amsthm}  
\usepackage{relsize}
\usepackage{algorithm}
\usepackage{algorithmic}

\def\aref#1{({\ref{#1}})}

\DeclareMathOperator*{\argmax}{argmax}

\DeclareMathOperator*{\minimize}{minimize}
\DeclareMathOperator*{\maximize}{maximize}

\DeclareMathOperator*{\subjectto}{subject~to}

\def\diag.{\mathop{\mathrm{diag.}}\nolimits}

\theoremstyle{plain} 

\newtheorem{theorem}{Theorem}

\newtheorem{lemma}{Lemma}
\newtheorem{problem}{Problem}

\def\proofname{Proof}
\Title{Automated Linear Function Submission-based Double Auction as Bottom-up Real-Time Pricing in a Regional Prosumers' Electricity Network}

\Author{Tadahiro Taniguchi $^{1}$*, Koki Kawasaki $^{2}$, Yoshiro Fukui$^{1}$, Tomohiro Takata$^{3}$%
, and Shiro Yano$^{4}$}

\address{%
$^{1}$ College of Information Science and Engineering, Ritsumeikan University, 1-1-1 Noji Higashi, Kusatsu, Shiga 525-8577, Japan\\
$^{2}$ Graduate School of Information Science and Engineering, Ritsumeikan University, 1-1-1 Noji Higashi, Kusatsu, Shiga 525-8577, Japan\\
$^{3}$ Research Organization of Science and Technology, Ritsumeikan University, 1-1-1 Noji Higashi, Kusatsu, Shiga 525-8577, Japan\\
$^{2}$ Division of Advanced Information Technology \& Computer Science, Tokyo University of Agriculture and Technology,
2-24-16 Nakamachi, Koganei, Tokyo, 184-8588, Japan}

\corres{e-mail: taniguchi@em.ci.ritsumei.ac.jp}

\abstract{A linear function submission-based double-auction (LFS-DA) mechanism for a regional electricity network is proposed in this paper. Each agent in the network is equipped with a battery and a generator. Each agent simultaneously becomes a producer and consumer of electricity, i.e., a {\it prosumer} and trades electricity in the regional market at a variable price. In the LFS-DA, each agent uses linear demand and supply functions when they submit bids and asks to an auctioneer in the regional market.
The LFS-DA can achieve an exact balance between electricity demand and supply for each time slot throughout the learning phase and was shown capable of solving the primal problem of maximizing the social welfare of the network without any central price setter, e.g., a utility or a large electricity company, in contrast with conventional real-time pricing (RTP). This paper presents a clarification of the relationship between the RTP algorithm derived on the basis of a dual decomposition framework and LFS-DA. Specifically, we proved that the changes in the price profile of the LFS-DA mechanism are equal to those achieved by the RTP mechanism derived from the dual decomposition framework except for a constant factor.
}

\keyword{Multi-agent system; Distributed algorithm; Dual decomposition; Double auction; Lagrangian relaxation}

\begin{document}
\nolinenumbers
{\bf NOMENCLATURE}

\vspace{4mm}
\begin{tabular}{lp{8.5cm}}\hline
        $l^{t+}_{i} \in [ l^{t+,\min}_{i}, \infty )$    &  Electric energy consumption profile\\
        $l^{t-}_{i} \in [ 0, l^{t-,\max}_{i} )$          & Electric energy generation profile\\
        $b^{t+}_{i} \in [ 0, b^{t+,\max}_{i} ]$          & Battery charge profile\\
        $b^{t-}_{i} \in [ 0, b^{t-,\max}_{i} ]$          & Battery discharge profile\\
        $m^{t+}_{i} \in [ 0, m^{t+,\max}_{i} ]$          & \shortstack{Profile of electric energy sold to the local electricity market} \\
        $m^{t-}_{i} \in [ 0, m^{t-,\max}_{i} ]$          & \shortstack{Profile of electric energy bought from the local electricity market} \\
        $g^{t+}_{i} \in [ 0, \infty )$                          &\shortstack{Profile of electric energy sold to the outside grid} \\
        $g^{t-}_{i} \in [ 0, g^{-,\max}_{i} ]$           &\shortstack{Profile of electric energy bought from the outside grid} \\
        $x^t_i$ & \shortstack{Profile of state vector}\\
$ s^t_i \in  [0 , s^{\max}_{i}]$ & \shortstack{Profile of the state of charge (SOC) of the battery} \\
$\eta_i \in [0,1]$ & \shortstack{Storage efficiency}  \\
 $\gamma \in [0,1]$ & \shortstack{Electricity transmission efficiency} \\
 $C^t_i$ & \shortstack{Cost function for generating electric energy}\\
 $D^t_i$ & \shortstack{Utility function for consuming electric energy}\\
 $\phi^t_i$ & \shortstack{Individual utility function}\\
 $W^t_i$ & \shortstack{Individual welfare function}\\
 $p_t$ & \shortstack{Price profile}\\
 $p^{G+}_t$ & \shortstack{Price of electricity sold to the outside grid}\\
 $p^{G-}_t$ & \shortstack{Price of electricity bought from the outside grid}\\
 $\alpha^t_i$ & \shortstack{Constant term of parameters of the bidding function}\\
 $\beta^t_i$ & \shortstack{Primary coefficient term of parameters of the bidding function}\\\hline
\end{tabular}

The superscript $t \in {\cal T}$ denotes the $t$-th time slot in a day. The subscript $i \in {\cal N}$
 denotes the $i$-th agent in an electricity network.
\section{INTRODUCTION}
\label{sec1}
In recent years, solar power and other renewable forms of energy, including wind and hydroelectric power, have been attracting attention because fossil fuel will become difficult to mine in the near future. In contrast with thermal power generation based on fossil fuel,
 electric power generated by photovoltaic (PV) cells and wind turbines is uncontrollable and unpredictable because it depends on the weather and climate. Therefore, adapting consumer demand to the fluctuating supply is a key issue when developing our future electricity network. Recently, a decentralized autonomous electricity network has been studied as a way to overcome this problem~\cite{vogt2010market}.

In this paper, we propose a double auction-based demand-side management mechanism for a regional prosumers' electricity network named the inter-intelligent renewable energy network (i-Rene)~\cite{Taniguchi2012}. A schematic view of i-Rene is shown in Figure~\ref{fig:i-rene}. In i-Rene, locally generated electric power is consumed in the local electricity network. Currently, surplus energy produced by PVs connected to the electricity network in Japan is sold at a fixed high price, i.e., the price is set as the Feed-in Tariff, and transferred to a grid managed by a regionally monopolizing electric power company, almost forcibly. This causes reverse power flow, which could destabilize the global electric power network.

Therefore, i-Rene is designed not to transfer surplus energy to the external grid, to which it is connected via a gateway, and which essentially transmits electric power in a single direction. Each house is equipped with a smart meter running intelligent software. The software automatically trades electricity in the regional electricity market and allocates the generated electricity appropriately by referring to the load profile of each house.
The regional electricity market is automatically managed by the gateway.
 We assume that the regional electricity network covers an entire town in which the number of houses is assumed to be between ten and several thousand.

\begin{figure}[t]
\begin{center}
\includegraphics[width=0.7\textwidth,keepaspectratio=true]{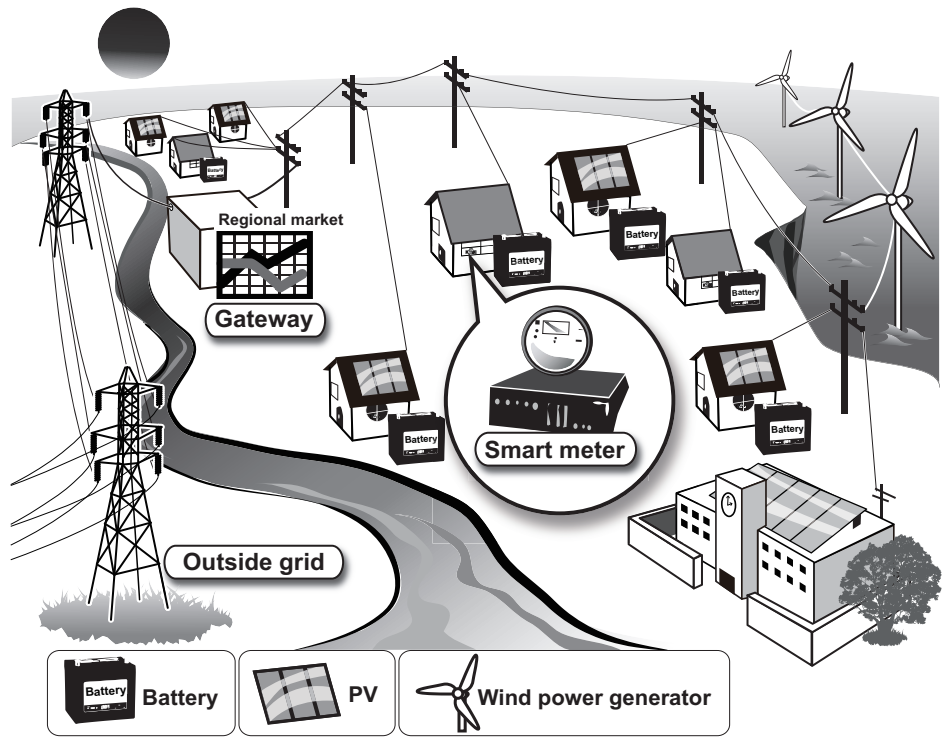}
\end{center}
\caption{Overview of inter-intelligent renewable energy network (i-Rene)}
\label{fig:i-rene}
\end{figure}

We assume that each house connected to i-Rene contains a generator using renewable energy resources and a battery. The intelligent software in the smart meter automatically trades electricity on behalf of the dweller.  However, local power generation, trading, and electricity stored in the batteries cannot always satisfy the total demand in the local grid. In such cases, consumers connected to i-Rene would have the option of buying electricity from an outside power grid through the gateway at a high price.
i-Rene is not centrally managed in that there is no central utility or large electricity company; instead, the agents in the network simply trade their electricity through double auction. Therefore, designing a double-auction mechanism and analyzing its properties are important.

Double-auction mechanisms have been used previously for designing multi-agent software systems to manage regional electricity networks~\cite{Kok2005,Hommelberg2007,Vytelingum10b,Saad2012,Taniguchi2012}. Double auction involves the submission of bids by potential buyers and the simultaneous submission of asks by potential sellers to an electricity market.
If an adequate price balancing demand and supply is determined, the auctioneer clears the market. The determined price allows the participants to buy or sell electricity on the basis of their submitted bids or asks.
One example of a double auction-based multi-agent software system for regional electricity network is the PowerMatcher~\cite{Kok2005,Hommelberg2007,kok2012}. Several field tests proved the ability of the PowerMatcher to allocate energy resources efficiently in a decentralized manner. 
In the PowerMatcher, each agent submits continuous demand and supply functions directly to an auctioneer in a market. Several multi-agent systems for realizing intelligent electricity networks employ such a function submission-based double auction (FS-DA)~\cite{Kok2005,Hommelberg2007,kok2012,Taniguchi2012,Li2009}.
The advantage of FS-DA is that the auctioneer determines an appropriate clearing price by calculating the intersection between aggregate demand and supply functions without iterative communication between the auctioneer and the agents. This calculation is easily performed using simple agent software without any manual intervention by humans.

On the other hand, a number of studies have been devoted to the use of real-time pricing (RTP) in smart grids~\cite{Saad2012,Alizadeh2012,Samadi2010,Li2011,Gatsis2011,Giannakis2011,Miyano2012,Disfani2014,GridWise1,Rad2010}. These studies mainly used game  theory and/or control theory to develop RTP algorithms and to prove their optimality and/or robustness. In particular, the dual decomposition framework provides a sophisticated explanation and theoretical foundation for RTP algorithms~\cite{Palomar2006}. However, the disadvantage of most of the dual decomposition-based RTP algorithms is that an iterative communication between a central utility, which settles the electricity price, and the agents are required to balance demand and supply. In addition, most of the dual decomposition-based RTP algorithms assume that a central utility ``does not'' behave selfishly, but behaves so as to optimize social welfare, i.e., it attempts to solve the master problem dedicatedly.

The use of the double-auction mechanism leads to the automatic generation of a time-varying  price profile by simply balancing demand and supply. The consumers and the power producers are expected to change their load, generation, and storage profiles by reacting to the time-varying price profile so as to maximize their profit.
This means the mechanism has an intrinsic RTP mechanism.
However, the optimality and the performance of LFS-DA in a regional prosumers' electricity network with distributed power resources and batteries have not been proved.

In this paper, we describe a linear function submission-based double auction (LFS-DA) mechanism for a day-ahead power market, in which an adequate price profile is gradually formed to control the demand and supply without any central price controllers. Throughout the learning phase, the LFS-DA is able to achieve an exact balance between electricity demand and supply for each time slot.
We show the LFS-DA to be essentially equal to dual decomposition-based RTP from the viewpoint of changes in the price profile. 
Especially, our main contributions are as follows.
\begin{itemize}
\item We propose an LFS-DA mechanism for a regional prosumers' electricity network. The mechanism is able to achieve an exact balance between electricity demand and supply at each moment during iteration.
\item We prove that the changes in the price profile obtained by the LFS-DA become the same as those obtained by the RTP derived from the dual decomposition framework except for a constant factor.
\end{itemize}

Our main result, i.e., Theorem 1, is very simple but powerful. It shows that simple LFS-DA is sufficient for managing the day-ahead market in i-Rene without the need for a central utility to perform RTP to balance demand and supply, and attempts to maximize social welfare by setting an appropriate price profile.

The remainder of this paper is organized as follows. Section~\ref{sec2} provides the background to our research and related work. Section~\ref{sec3} describes the problem definitions and basic assumptions of the target electricity network used in our research, which is a regional electricity network consisting of many prosumers. It also describes an RTP algorithm derived from the dual decomposition framework. Section~\ref{sec5} describes our proposed LFS-DA mechanism and provides proof that the LFS-DA is intrinsically equivalent to dual decomposition-based RTP. Section~\ref{sec6} contains details of a simulation experiment and its results. Finally, Section~\ref{sec7} concludes this paper.

\begin{figure}[tbp]
 \begin{minipage}{0.5\hsize}
\centering
\includegraphics[width=\textwidth,keepaspectratio=true]{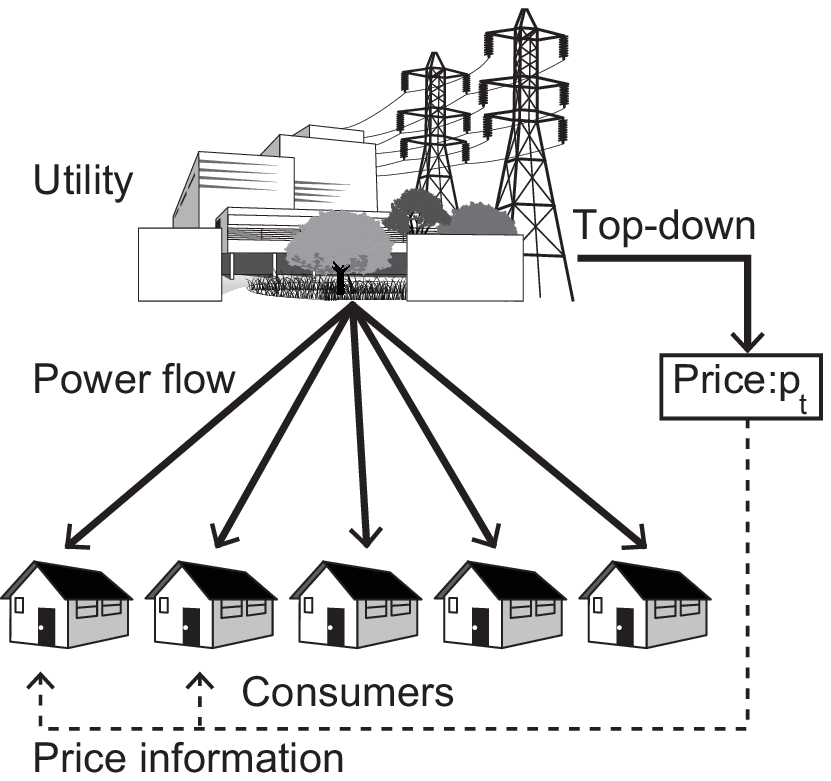}
 \end{minipage}
 \begin{minipage}{0.5\hsize}
 \centering
\includegraphics[width=\textwidth,keepaspectratio=true]{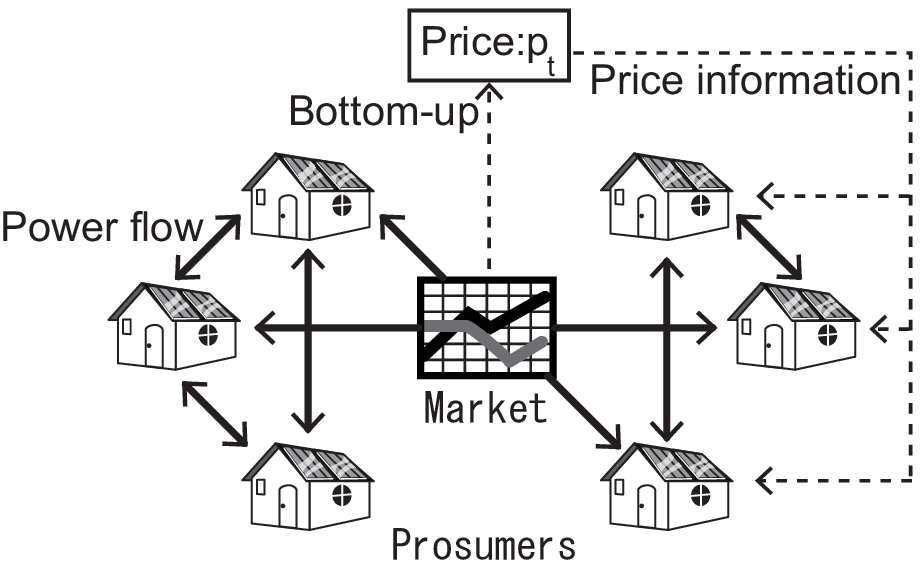}
 \end{minipage}
 \caption{(Left) real-time pricing in a conventional power grid, and (right) double-auction pricing mechanism in a regional prosumers' electricity network}\label{fig:rtp-da}
\end{figure}

%

\section{BACKGROUND AND RELATED WORK}\label{sec2}
\subsection{Penetration of distributed energy resources and smart grids}
Distributed energy resources (DER), e.g., distributed energy generators and storage facilities,  have gradually been introduced into households, plants, and offices. Such distributed generation (DG) is mostly performed by using renewable energy resources, including solar, wind, and hydroelectric power.
Conventionally, especially in Japan, large power generation companies have generated electricity centrally with thermal electric power stations and nuclear power plants, which they then transmit to consumers on a unilateral basis. The situation differs from country to country. Specifically, in Japan, each power supplier has traditionally been a regional monopolizer, obliged by law to completely fulfill consumers' electricity demands.

 However, the increasing penetration of DERs and the emergence of smart grid technologies are changing the requirements of power networks.
The wide-scale introduction of DERs causes a reverse power flow to conventional power grids, which would be rendered increasingly unstable and confused.
Therefore, it is important to develop a regional electricity network capable of mitigating the effects of a reverse power flow.
In this context, several power network concepts, including a smart grid~\cite{Brown2008}, micro grid~\cite{Dou2013}, and a virtual power plant (VPP) ~\cite{Molderink2010}, have been studied with the aim of smoothly introducing DERs to the grid.

In the development of a regional network based on DERs, demand side management (DSM) becomes an important design element.
However, in an electricity network, balancing the demand and supply at each moment is crucially important.
Therefore, achieving a balance between the demand and supply requires the power demand, rather than the power supply, to be controlled. 
A comprehensive survey of research relating to DSM was conducted by Alizadeh et al. \cite{Alizadeh2012}. Moreover, many field tests demonstrating the effectiveness of DSM have already been performed, e.g., \cite{GridWise1,GridWise2,kok2012}.

\subsection{Demand side management}
The realization of efficient, robust, and applicable DSM has led to the study of market-based methods~\cite{Alizadeh2012}, which have used price-based incentives, e.g., time-of-use (TOU) pricing, RTP, and peak-time rebates, to control consumers' electricity demand.
The approaches used in game and control theory have provided the theoretical foundations for DSM methodologies.

A comprehensive survey of the application of game theory to smart grid systems was provided by Saad et al.~\cite{Saad2012}. Vytelingum et al. analyzed the Nash equilibrium for an electricity grid to which distributed micro-storage devices are connected. Their multi-agent simulation showed that this had the effect of flattening the electricity price profile~\cite{Vytelingum10}.
Mohsenian-Rad et al. proposed a distributed algorithm for autonomous DSM~\cite{Rad2010}. They provide analytical proof to show that their proposed algorithm guarantees a significant reduction in the total energy cost of the system.

Control theory forms the theoretical foundation of RTP on the basis of a {\it dual decomposition}  framework~\cite{Samadi2010,Li2011,Gatsis2011,Giannakis2011,Miyano2012,Disfani2014}.
A comprehensive tutorial of dual decomposition was compiled by Palomar et al.~\cite{Palomar2006}.
In a dual decomposition-based approach, the problem of maximizing the social welfare of consumers and/or a power producer under several constraints is defined as the primal problem. This problem is transformed into a dual problem by applying a Lagrange relaxation and decomposition methods which have historically been used in optimization problems in electric power systems~\cite{Bard1988}.
 This dual problem can be divided into many sub-problems and a master problem. Each sub-problem corresponds to a selfish profit maximization problem of a consumer or producer, whereas the master problem corresponds to the problem involving a search for the optimal price profile. In a dual decomposition approach, the introduced Lagrange multipliers are interpreted as ``prices''. In this way, the dual decomposition method provides a mathematical foundation based on RTP and derives each pricing algorithm of DSM depending on each condition of the target electric power network. On the basis of the dual decomposition framework and related concepts, various RTP methods have been proposed~\cite{Samadi2010,Li2011,Gatsis2011,Giannakis2011,Miyano2012,Disfani2014,Chang2014,Yo2013,Kiani2012,Zhao2013a}.

\subsection{Multi-agent system and double auction}
Most of the previous studies relating to RTP control theory assume that {\it power producers} and {\it consumers} exist separately in an electric power network as shown in Figure~\ref{fig:rtp-da} (left). In contrast, in our work, we focus on an electricity network consisting of electricity {\it prosumers} as shown in Figure~\ref{fig:rtp-da} (right). The term prosumer refers to a person or entity who is both a producer and a consumer at the same time~\cite{wikinomics}. If a consumer installs a PV system at his/her house, he/she not only becomes a consumer, but also a power producer. 
The economical structure of an energy network in which regional prosumers participate would differ completely from that of a conventional unilateral power grid. Such a network has been studied in the context of multi-agent system-based microgrid.

In recent times, multi-agent systems using the double-auction method for managing electric power networks have been gaining attention~\cite{Zidan2012,Dou2013}.
Kok et al. developed the PowerMatcher Smart Grid Technology~\cite{kok2012}, a multi-agent-based distributed software system for the market integration of small and medium sized DER units.
In a PowerMatcher cluster, many agents interact with each other and the equilibrium price for electricity is determined automatically by using the double auction approach. The effectiveness of dynamic pricing using the PowerMatcher has been tested in simulations and field experiments~\cite{Kok2005,Hommelberg2007,kok2012}.
Vytelingum et al. proposed a market-based mechanism on the basis of Continuous Double Auction (CDA), and the use of an efficient trading agent strategy. They showed that the mechanism and the agents were able to cope with unforeseen demand or increased supply capacity in real time~\cite{Vytelingum10b}.  Taniguchi et al. proposed an inter-intelligent renewable energy network (i-Rene), which is automatically managed by using double auction~\cite{Taniguchi2012}.
Related papers reported the capability of their methods to achieve a reduction in peak consumption and/or the efficient utilization of energy~\cite{kok2012,Vytelingum10b,Taniguchi2012}.
However, they were unable to prove that the double auction-based approaches maximize the social welfare of the prosumers' network from an analytical point of view. This contrasts with the previously mentioned dual decomposition-based methods, all of which have an analytical foundation that guarantees an increase in the social welfare. 

Most of the multi-agent-based approaches employ double auction in the regional electric power network to determine the price and amount of electricity transacted. In double auction, bids (i.e., buy orders) and asks (i.e., sell orders) have to be matched to enable a certain amount of electricity to be transacted and the price to be determined. However, if various agents submit bids and asks at various prices for various units of electricity, the market auctioneer has to use a large number of iterative communications to match demand and supply.
To reduce the communication overhead, Kok et al. simplified the communication by using only demand functions and prices~\cite{kok2012}.
Agents in the network submit individual supply and/or demand functions to the market and if all of the agents submit continuous functions, the market is able to calculate an adequate clearing price and suitable transactions by searching for the point at which the aggregate demand and supply functions intersect. In this paper, we refer to this double-auction method as the function submission-based double-auction (FS-DA) mechanism. In the Agent-based Modeling of Electricity Systems (AMES) test bed model developed by Li et al. and in the i-Rene developed by Taniguchi et al., the submitted functions are also restricted to linear functions~\cite{Li2009,Taniguchi2012}. We refer to an FS-DA mechanism that only uses linear functions as an LFS-DA mechanism.

In contrast with the traditional top-down approach, in which price profiles are determined by a public utility or by RTP generation (see Figure~\ref{fig:rtp-da} left), our double-auction approach determines price profiles dynamically through interaction between sellers and buyers in a bottom-up manner (see Figure~\ref{fig:rtp-da} right).
%
Several related mechanisms were proposed. Sam et al. proposed a function bidding mechanism for multi-agent charging of electric vehicles~\cite{Sam2014}. Their mechanism could eliminate an iterative exchange of messages by introducing function bidding.
Kiani et al. proposed a dynamic market mechanism and analyzed the market as a dynamical system~\cite{Kiani2014,Hansen2014}. Xu et al. proposed and analyzed a simple market mechanism using supply function bidding~\cite{Xu2015}. Their supply function differed from that of our model.
The research most closely related to ours was conducted by Papavasiliou et al., who proposed a Newton algorithm-based double-auction mechanism~\cite{Papavasiliou2010}. Their mechanism, supply function bidding, is interpreted as a Newton algorithm for optimizing problems with a decomposable structure. The mechanism enables an automated demand response control system through a double auction-like mechanism.
However, they did not consider a prosumer network in which each consumer uses batteries and renewable energy resources, and did not demonstrate the mathematical relationship between dual decomposition-based RTP and the function submission-based double-auction mechanism under the condition that each agent behaves selfishly, i.e., they required each agent to use a Newton algorithm to update its variables. In this paper, we show that the price change in LFS-DA becomes the same as dual decomposition-based RTP except for a constant factor, which is our main theorem.

\section{PROBLEM DEFINITION}\label{sec3}

\begin{figure}[t]
\centering
\includegraphics[width=0.7\textwidth,keepaspectratio=true]{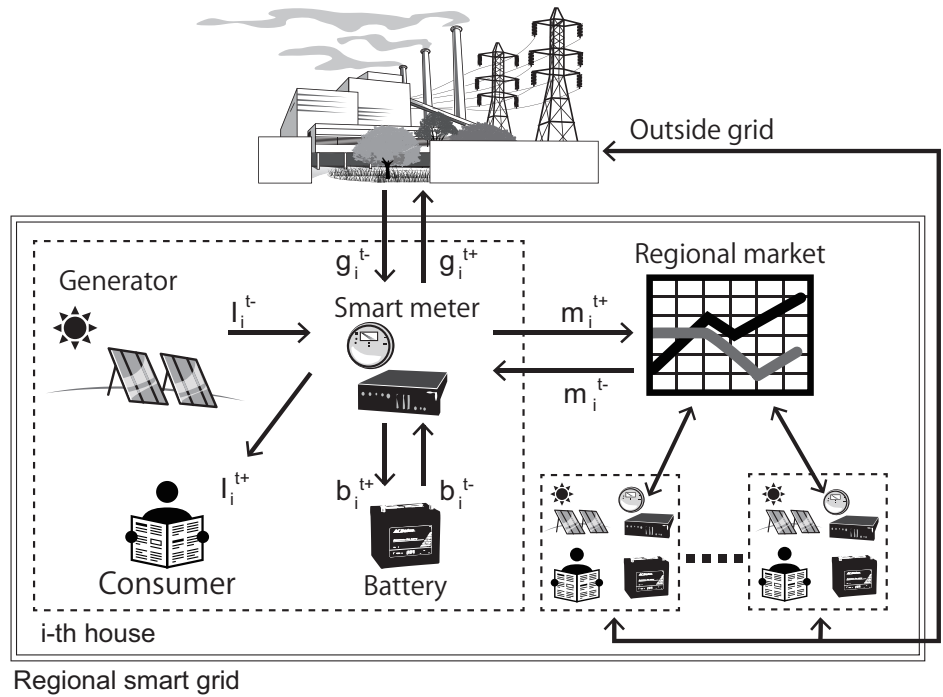}
\caption{Power flow in the inter-intelligent renewable energy network (i-Rene)}
\label{fig:irene}
\end{figure}

\subsection{Basic assumptions}
In this work, we define a target regional electricity network as follows (a schematic diagram is shown in Figure~\ref{fig:i-rene} and Figure~\ref{fig:irene}):
we do not distinguish between electricity suppliers and consumers; rather, we assume that all people on the grid are capable of producing and consuming electricity as {\it prosumers}. Thus, it is assumed that each house is equipped with a generator powered by renewable energy, e.g., PV cells or wind power generators, and a storage device, e.g., a battery.
Furthermore, each house is considered to have a smart meter running software, which automatically controls the electricity of the house by managing its battery, transmitting electricity to other houses, and communicating with other information systems.

In the regional electricity network, houses are connected through a regional electricity market through which each prosumer can sell and buy electricity.
The local electricity price fluctuates on the basis of regional demand and supply.
The prosumer attempts to optimize its trading rule, consumption, and generation so as to maximize his/her welfare.
Basically, i-Rene is intended to behave as an independent electricity network.

The outside grid provides the prosumer with an optional alternative from which to buy electricity at the fixed high price at any time. In addition, the prosumer can also sell surplus electricity to the outside grid at the fixed low price, which is set so as to inhibit reverse power flow.
We assume the outside grid is a unilateral conventional grid, as is common in Japan.
Hereafter, we refer to a household, including the prosumers living in the house and the software, simply as an {\it agent}.

Suppose there are $N$ agents in i-Rene and a set of these agents is represented by ${\cal N} := \{1,2,...,N\}$.
An agent can consume, generate, charge, discharge, buy, and sell electricity through its smart meter during every one of the $T$ time slots. In addition, the number of time slots for transactions is shared by all of the agents in the grid. The set of time slots is defined as ${\cal T} := \{1,2,...,T\}$.
The $i$-th agent can determine the amount it consumes $l^{t+}_{i}$, generates $l^{t-}_{i}$, charges $b^{t+}_{i}$, discharges $b^{t-}_{i}$, sells to $m^{t+}_{i}$, and buys from the regional market $m^{t-}_{i}$, and sells to $g^{t+}_{i}$, and buys from the outside grid $g^{t-}_{i}$ during each time slot $t$.
The variables of the $i$-th agent $(i\in{\cal N})$ at time $t\in{\cal T}$ are defined as shown under
Nomenclature. The relationship among the variables is also schematically shown in Figure~\ref{fig:irene}.

Each variable has its individual lower limit and/or upper limit owing to the limited capability of each apparatus or system. The superscripts $\cdot ^+$ and $\cdot ^-$ represent the direction from the viewpoint of a smart meter, which is at the center of electric energy flows in an agent's house. Specifically, $\cdot ^+$ and $\cdot ^-$ represent outflow and inflow, respectively.
We assume that the energy flow through the smart meter adheres to the law of the conservation of energy for each time slot $t$ as follows.
    \begin{align}
      l^{t+}_{i} - l^{t-}_{i} + b^{t+}_{i} - b^{t-}_{i} + m^{t+}_{i} - m^{t-}_{i} + g^{t+}_{i} - g^{t-}_{i} = 0 
    \end{align}
The amount of electricity demanded has to be balanced against the supply in the regional electricity network for each time slot $t$. We assume
    \begin{align}
      \sum_{i\in {\cal N}} (\gamma m^{t+}_{i} - m^{t-}_{i}) = 0 \label{eq:balancing}~~~\forall t \in {\cal T}.
    \end{align}
where $\gamma \in [0,1]$ represents the electricity transmission efficiency. If $\gamma = 1$, there is no electricity energy loss during transmission. In addition, the storage efficiency $\eta_i \in [0,1]$ must be taken into consideration. If $\eta_i  = 1$, the charged electricity can be fully charged without any loss.
The storage profile $s^{t}_{i}$ represents the state of charge (SOC) of the $i$-th agent's storage device at time $t$ and is expressed by the following equation.
    \begin{align}
      s^{t}_{i}
      &:= s^{t-1}_{i}+ \eta_i b^{t+}_{i} - b^{t-}_{i} ~~~~~\forall i\in{\cal N}, \forall t\in{\cal T}\\
      &=  s^{\text{init}}_{i} + \sum_{k\in \{1,2, \ldots, t \}}(\eta_i b^{k+}_{i} - b^{k-}_{i}).
    \end{align}
where $s^{\text{init}}_{i}$ is the initial SOC of the $i$-th agent's battery.

Agents can buy electricity from and sell electricity to the outside grid at the price of $p^{G-}_{t}$ and $p^{G+}_{t}$ per unit, respectively. The situation for which reverse power flow is completely prohibited is modeled by setting $p^{G+}_{t}=0$. An immediate resale behavior between an agent and the outside grid is suppressed by assuming the following constraint.
    \begin{align}
      0 \le p^{G+}_{t} \le p^{G-}_{t}.
    \end{align}

\subsection{Social welfare}
We define the cost for generating electric energy $l^{t+}_{i}$ for the $i$-th agent by $C^t_i:\mathbb{R} \rightarrow \mathbb{R}$, where $C^t_i$ is a convex function of class $C^2$ and
define the utility for consuming electric energy $l^{t-}_{i}$ for the $i$-th agent by $D^t_i:\mathbb{R} \rightarrow \mathbb{R}$, where $D^t_i$ is a concave function of class $C^2$.
The welfare $W_i:\mathbb{R}^{8T}\times \mathbb{R}^{T} \rightarrow \mathbb{R}$ of the $i$-th agent is defined as follows.
\begin{align}
        W_i(x_i,p) &:=\sum_{t \in {\cal T}} W^{t}_{i}(x^{t}_{i},p_t), \label{utility_i}\\
        W^t_i(x^{t}_{i},p_t) &:= \phi^t_i(x^t_i) - p_t \gamma m^{t+}_{i} + p_t m^{t-}_{i}, \label{eq:welf1}\\
        \phi_i^t(x^t_i) &:= D^{t}_{i}(l^{t+}_{i}) - C^t_i( l^{t-}_{i} ) + p^{G+}_{t}g^{t+}_{i} - p^{G-}_{t}g^{t-}_{i}\\
        x_{i} &:= ( x^{t}_{i} )_{t \in {\cal T}}, \\
        x^{t}_{i} &:= ( l^{t+}_{i}, l^{t-}_{i}, b^{t+}_{i}, b^{t-}_{i}, m^{t+}_{i}, m^{t-}_{i}, g^{t+}_{i}, g^{t-}_{i} ),
      \end{align}
where $p = (p_t)_{t\in{\cal T}}:=(p_1, ..., p_T)$, $p_t$ is a buyer's price profile in the regional electricity market at time $t\in{\cal T}$, and $\phi^t_i(x_t)$ is a utility function for the $i$-th agent at time $t$ including payments to an outside grid. Because of the electricity losses in  transmission, the seller's price and buyer's price become different. To balance the amount of money paid by sellers and buyers, the seller's price becomes $\gamma p_t$.
The $i$-th agents behavior at time $t$ is represented by $x^t_i$, concisely. We call $x^t_i$ a state vector for the $i$-th agent at time $t$.
Next, we define the social welfare of the network by $W(x,p)$.
    \begin{align}
 W(x,p) :=&      \sum_{i\in{\cal N}} \sum_{t\in{\cal T}} \big( D^{t}_{i}(l^{t+}_{i}) - C^t_{i}(l^{t-}_{i}) + p^{G+}_{t}g^{t+}_{i} - p^{G-}_{t}g^{t-}_{i} \nonumber \\
 &~~~~~~~~~~ + p_t \gamma m^{t+}_{i} - p_t m^{t-}_{i} \big) \label{eq:sw1}\\
   =&  \sum_{t \in {\cal T}}\big(\sum_{i\in{\cal N}} \phi^t_i(x_i) +  p_t \underbrace{\sum_{i\in {\cal N}} (\gamma m^{t+}_{i} - m^{t-}_{i})}_{=0\ \text{from equation \aref{eq:balancing} }} \big)  \label{eq:sw2}
   \end{align}
where  $x = (x_i)_{i\in{\cal N}}:=(x_1, ..., x_N)$.
This shows that the social welfare inside of the regional prosumers' electricity network does not depend on the price profile of the internal regional market.
Therefore, maximizing $W(x,p)$ is the same as maximizing $ \phi(x) :=\sum_{i \in {\cal N}} \phi_i (x_i) $ where $\phi(x)$ is the sum of the individual utility functions $ \phi_i(x_i) :=\sum_{t \in {\cal T}} \phi^t_i (x^t_i) $.

\subsection{Real-Time Pricing Algorithm}
The following problem is defined as the {\it primal problem}, which has the purpose of maximizing the social welfare of the regional network.

\begin{problem}[Primal problem]
\label{prob:PP}
  \begin{align}
    \underset{x\in\mathbb{R}^{8NT}}{\maximize} &\text{ } \sum_{i \in {\cal N}} \phi_i (x_i), \label{primal1}\\
    \subjectto
    &\text{ } x_i \in {\cal X}_i ~~~\forall i \in {\cal N}, \label{primal2}\\
    &\text{ } \sum_{i\in{\cal N}} f^{t}_{i}(x^t_{i}) = 0 ~~~\forall t \in {\cal T}, \label{primal3}\\
    &~~f^{t}_{i}(x^t_i) := \gamma m^{t+}_{i} - m^{t-}_{i} .\label{primal4}
  \end{align}
\end{problem}
 where ${\cal X}_i$ is a feasible set for the $i$-th agent except for the network constraints (see APPENDIX~\ref{ap_sec1}).

The primal problem is numerically solved by a solver. However, no electric power companies, utilities, or governments can directly control the behavior of any of the agents.
In addition, they cannot obtain information about $C^t_i$ and $D^t_i$  because this information is usually private. Even if the central control system could obtain the private information, solving the primal problem requires a huge computational cost when the number of agents $N$ becomes large. In reality, solving the primal problem centrally is impossibly difficult in the current economic environment.
Therefore, a decentralized optimization method using price information is gathering attention in the context of smart grid management.

Based on the dual decomposition framework, an RTP mechanism maximizing social welfare, in which each agent maximizes its  welfare selfishly, can be derived, and the price profile is updated centrally using a sub-gradient.  The RTP algorithm is outlined in Algorithm~\ref{alg:dual}.

\begin{algorithm}[t]
\caption{Iterative solution of dual problem: real-time pricing (RTP) algorithm}
\label{alg:dual}
\begin{algorithmic}
\STATE $k \leftarrow  0$
\STATE Initialize the price profile $p^{(k)} = (p_t^{(k)})_{t \in {\cal T}}$
\REPEAT
\STATE // Each agent solves its sub-problem~\aref{dual_sub}, and obtains its solution.
\STATE Update $x^{(k+1)}_{i} \leftarrow  x^{*}_{i}(p^{(k)})$.
\STATE // A central utility solves the master problem~\aref{eq:dual_master} under the condition that $x_i = x_i^{(k+1)}$. 
\STATE Update $p^{(k+1)}_t \leftarrow p^{(k)}_t -\theta^t_k \xi_t(p^{(k)})$, for each $t$.
\STATE $k \leftarrow  k+1$
\UNTIL a predefined stopping criterion is satisfied.
\RETURN 
Transact $\{x_i^{(k)} \}_{i \in {\cal N}}$ with $p^{(k)}$ as price profile
\end{algorithmic}
\end{algorithm}\

The price profile $p^{(k)}_t$ is updated by using a known sub-gradient method~\cite{Palomar2006,Samadi2010}
Adopting this method, $p_t^{(k)}$ is updated as follows:
    \begin{align}
      p^{(k+1)}_t &= p^{(k)}_t - \theta^t_k \xi_t(p^{(k)}) \label{eq:dual_update}\\
      \xi_t(p^{(k)})
      &:=  \sum_{i\in{\cal N}} f^{t}_{i}( x^{*}_{i}(p^{(k)}) ) \\
      &= \sum_{i\in{\cal N}} \big(\gamma m^{*t+}_{i}(p^{(k)}) - m^{*t-}_{i}(p^{(k)}) \big) ,
    \end{align}
where $\theta^t_k > 0$ is the learning rate of the sub-gradient method.
Iterative updating of $x_i$ and $p$ provides the numerical solution of the dual problem (see~\ref{ap_sec1}). This solution is also a solution of the primal problem~\aref{primal1} - \aref{primal4} . Therefore, the social welfare of the network is expected to be maximized.
Further details are provided in APPENDIX~\ref{ap_sec1}.
From the viewpoint of the pricing mechanism, solving the dual problem equates to the determination of a price profile by a utility, so as to maximize the social welfare.
Therefore, we refer to this algorithm as an RTP algorithm.

Although this algorithm presents a feasible solution for the decentralized energy dispatch problem, it is unable
to achieve an exact balance between demand and supply before converging to an optimal solution, despite the necessity of such a balance in an electricity network. In each step, $\delta m^t :=  \Big(\sum_{i\in{\cal N}} \big(\gamma m^{*t+}_{i}(p^{(k)}) - m^{*t-}_{i}(p^{(k)})\big) \Big)_{t\in{\cal T}}$ is not exactly a zero vector in most cases.
Therefore, a heuristic process should be applied when aiming to compensate the difference $\delta m^t$. One of the simplest methods is that the utility compensates this difference when it is not zero by selling or buying $\delta m^t$ to or from the outside grid at the price of $p^{G+}_t$ or $p^{G-}_t$, respectively. In this work, we assume that the i-Rene gateway employs this heuristic procedure if i-Rene adopts the RTP algorithm.

In practical situations, the iteration cannot be performed many times, and only a small number of iterations are permitted. In such cases, the effect of $\delta m^t$ cannot be ignored.
In this context, a mechanism which would be able to achieve an exact balance between the demand and supply is desirable. The LFS-DA is such a mechanism as it has the ability to balance demand and supply exactly for each time slot and every iteration step, and is guaranteed to solve the problem in the same way as the dual decomposition-based RTP algorithm.

\section{THE LINEAR FUNCTION SUBMISSION-BASED DOUBLE AUCTION}\label{sec5}
The LFS-DA mechanism works as follows. A prosumer submits their individual linear demand and supply functions parametrized by two parameters, i.e., $\alpha^{t}_{i}$ and  $\beta^{t}_{i} $, described in \aref{submit_m1} and \aref{submit_m2}. Then, the market-clearing price $p_t$ is determined exactly by calculating the point at which the aggregate demand and supply functions intersect. After several iterations between an auctioneer and the prosumers, all prosumers transact their electricity using the determined clearing price.
In this section, we formulate the LFS-DA mechanism and disclose its theoretical properties.

\subsection{Transaction with LFS-DA}
In LFS-DA, each agent has a linear demand function $m^{t-}_i = \mu^{t-}_i(p_t)$ and a linear supply function $m^{t+}_i=\mu^{t+}_i(p_t)$ for each time slot $t$. These two functions are determined by using the two parameters $\alpha^{t}_{i}$ and  $\beta^{t}_{i} $.
    \begin{align}
      m^{t+}_i= \mu^{t+}_{i}(p_t) &:= \lfloor \beta^{t}_{i}p_t - \alpha^{t}_{i} \rfloor, \label{submit_m1} \\
     m^{t-}_i= \mu^{t-}_{i}(p_t) &:= \lfloor-\beta^{t}_{i}p_t + \alpha^{t}_{i} \rfloor, \label{submit_m2}
    \end{align}
where $\lfloor x \rfloor := \max(x,0)$. In the LFS-DA, each agent submits the two parameters, i.e., $\alpha^{t}_{i}$ and  $\beta^{t}_{i} $, to the  auctioneer in the regional market to provide information about the linear demand and supply functions.
In this work, we assume $\beta^{t}_{i} $ is a fixed positive constant, and $\alpha^{t}_{i}$ is a flexible variable that is optimized by each agent.

When each agent submits the two parameters $( \alpha^t_i, \beta^t_i)$
 to the market, an adequate price $p_t$ is determined that would suffice to clear the market by searching for the point at which the aggregate demand and supply functions intersect.
The constraint for balancing demand and supply \aref{eq:balancing} becomes
    \begin{align}
      \sum_{i\in{\cal N}} f^{t}_{i}(x_{i},p_{t}) &=
      \gamma \sum_{i\in{\cal N}} \lfloor \beta^{t}_{i}p_{t} - \alpha^{t}_{i}\rfloor - \sum_{i\in{\cal N}} \lfloor-\beta^{t}_{i}p_{t} + \alpha^{t}_{i} \rfloor \nonumber \\
      &=0 .\label{eq:submit_st_loss}
    \end{align}

The price $p_t$ divides the agents into two groups. If for the $i$-th agent, $\frac{ \alpha^{t}_{i} }{ \beta^{t}_{i} } \le p_t $, the agent becomes a seller, i.e., $\beta^{t}_{i}p_{t} - \alpha^{t}_{i}= m^{t+}_{i} > 0$ and $m^{t-}_{i} = 0$. In contrast, if for the $i$-th agent $\frac{ \alpha^{t}_{i} }{ \beta^{t}_{i} } > p_t$, the agent becomes a buyer, i.e., $m^{t+}_{i}=0$ and $ -\beta^{t}_{i}p_t + \alpha^{t}_{i} = m^{t-}_{i} > 0 $. We define the set of suppliers at time $t$ by $I_t^{+}(p_t) := \{ i \mid  \frac{ \alpha^{t}_{i} }{ \beta^{t}_{i} } \le p_t , i \in {\cal N} \}$, and the set of consumers  at time $t$ by $I_t^{-}(p_t) := \{ i \mid  \frac{ \alpha^{t}_{i} }{ \beta^{t}_{i} } > p_t , i \in {\cal N}  \}$.
In this case, on the basis of the constraint ~\aref{eq:submit_st_loss}, an appropriate clearing price  has to satisfy the following equation.
    \begin{align}
    \aref{eq:submit_st_loss}  \Leftrightarrow
      0&= \gamma \sum_{i\in I_t^{+}(p_{t})} (\beta^{t}_{i}p_{t} - \alpha^{t}_{i}) - \sum_{i\in I_t^{-}(p_{t})} (-\beta^{t}_{i}p_{t} + \alpha^{t}_{i})  \label{submit_st_loss2}
      \end{align}
By solving the above equation for $p_t$, we obtain
     \begin{align}
\aref{submit_st_loss2} \Leftrightarrow
   p_{t} &= \frac
      { \gamma \alpha^{t}_{I_t^{+}(p_{t})} + \alpha^{t}_{I_t^{-}(p_{t})} }
      { \gamma \beta^{t}_{I_t^{+}(p_{t})} + \beta^{t}_{I_t^{-}(p_{t})} }, \label{eq:pt_balancing}
    \end{align}
where $ \alpha^{t}_{I_t^{+}(p_t)} :=\sum_{i\in I_t^{+}(p_t)} \alpha^{t}_{i}$,
$ \alpha^{t}_{I_t^{-}(p_t)} :=\sum_{i\in I_t^{-}(p_t)} \alpha^{t}_{i}$,
$\beta^{t}_{I_t^{+}(p_t)} :=\sum_{i\in I_t^{+}(p_t)} \beta^{t}_{i}$,
 and $\beta^{t}_{I_t^{-}(p_t)} :=\sum_{i\in I_t^{-}(p_t)} \beta^{t}_{i}$.
The price $p_{t}$ satisfying \aref{eq:pt_balancing} exactly balances demand and supply, and fulfills the constraint in~\aref{eq:submit_st_loss}.

\begin{lemma}[The clearing price]\label{lemma0}
The clearing price $p_t$ that satisfies \aref{eq:pt_balancing} uniquely exists and be calculated exactly.
\end{lemma}
\begin{proof}	
$  \gamma  \lfloor \beta^{t}_{i}p_{t} - \alpha^{t}_{i}\rfloor - \lfloor-\beta^{t}_{i}p_{t} + \alpha^{t}_{i} \rfloor $ is a piecewise linear and monotonically increasing function of $p_{t}$ and its value at $p_t =0 $ is negative.
Therefore $ \gamma \sum_{i\in{\cal N}} \lfloor \beta^{t}_{i}p_{t} - \alpha^{t}_{i}\rfloor - \sum_{i\in{\cal N}} \lfloor-\beta^{t}_{i}p_{t} + \alpha^{t}_{i} \rfloor$ is also a piecewise linear and monotonically increasing function of $p_{t}$ and its value at $p_t =0 $ is negative. Therefore, the solution of \aref{eq:submit_st_loss} uniquely exists and can be calculated exactly.
\end{proof}

In the LFS-DA, when the auctioneer in the market receives $(\alpha_i^t ,\beta_i^t )_{i \in {\cal N}, t \in {\cal T}}$ from all of the agents at the $k$-th iteration, the  auctioneer determines the price $p^{(k+1)}_{t}$ by solving~\aref{eq:pt_balancing}. Based on the price $p^{(k+1)}_t$, each agent either transmits $m^{(k+1)t+}_i=\mu^{t+}_i(p^{(k+1)}_t)$ or receives $m^{(k+1)t-}_i=\mu^{t-}_i(p^{(k+1)}_t)$. After determining the amount of  $(m^{(k+1)t+}_i, m^{(k+1)t-}_i )$, each agent has to reconfigure its $x^t_i$ for the given $(m^{(k+1)t+}_i, m^{(k+1)t-}_i )$ so as to satisfy the constraint $x_i \in {\cal X}_i$.
Under the condition that each agent behaves rationally, i.e., selfishly, the process of the LFS-DA mechanism is described as Algorithm \ref{alg:LFS-DA}.

\subsection{LFS-DA as an bottom-up RTP method}
In the same way as RTP, communications between an auctioneer and prosumers are required by the LFS-DA to maximize social welfare.
When a price profile $(p^{}_t )_{t \in \mathcal{T}}$ is observed, each rational agent tries to maximize its welfare by optimizing $\alpha^t_i$ on the basis of the announced price profile $(p^{}_t )_{t \in \mathcal{T}}$\footnote{Note that $\beta^t_i$ is fixed in this work.}.

When $\beta^t_i$ and $p_t$ are fixed, $\alpha^i_t$ can be determined uniquely from  $(m^{i+}_t,m^{i-}_t)$ as follows.
\begin{align}
\alpha^{t}_{i} &= \beta^{t}_{i}p^{}_t + (m^{t-}_i - m^{t+}_i).\label{eq:from_m2a}
\end{align}
Equations \aref{submit_m1}, \aref{submit_m2}, and  \aref{eq:from_m2a} show that $\alpha^{t}_{i}$ and
 $(m^{i+}_t,m^{i-}_t)$ have bijective relation under the condition that $\beta^t_i$ and $p_t$ are fixed.
Therefore, the optimal solution $\alpha^{*t}_i(p^{(k)}_t)$ can be obtained using the following procedure.
    \begin{align}
      x^{*}_{i}(p^{(k)}) &= \underset{ x_i \in {\cal X}_i}{ \argmax } \text{ }\left( \phi_i (x_i) + \sum_{t\in{\cal T}}p^{(k)}_t (\gamma m^{t+}_{i} - m^{t-}_{i}) \right) \label{eq:lfs-da-sub}\\
      \alpha^{*t}_{i}(p^{(k)}_t) &= \beta^{t}_{i}p^{(k)}_t + (m^{*t-}_i(p^{(k)}_t) - m^{*t+}_i(p^{(k)}_t)).\label{eq:from_m2a_opt}
    \end{align}
   where $x^{*t}_{i} = ( l^{*t+}_{i}, l^{*t-}_{i}, b^{*t+}_{i}, b^{*t-}_{i}, m^{*t+}_{i}, m^{*t-}_{i}, g^{*t+}_{i}, g^{*t-}_{i} )$ and $x^{*}_{i} = (x^{*t}_{i})_{t \in {\cal T}}$ .
Equation \aref{eq:lfs-da-sub} corresponds to an optimal solution of the sub-problem of the dual problem introduced in the APPENDIX~\ref{ap_sec1}.

Here, we formalize each agent's iterative learning process.
The change in the price profile and agents' behavior are determined as shown in Algorithm \ref{alg:LFS-DA}.
In the $k$-th iteration, the price profile is assumed to be $p^{(k)} := ( p^{(k)}_t )_{t \in {\cal T}}$. Each agent selfishly maximizes its welfare and obtains $x^{*}_i$, before obtaining $\alpha^{*t}_{i} = (\alpha^{*t}_{i})_{t \in {\cal T}}$.
The clearing price settled by $\text{\bf market\_clearing}$ can be obtained by solving equation~\aref{eq:pt_balancing}.

On the basis of the new price profile $p^{(k+1)}$, the amount of electricity for transacting through the regional electricity market $(m^{t+}_i, m^{t-}_i) $ is updated as follows.
\begin{align}
(m^{(k+1)t+}_i, m^{(k+1)t-}_i) \leftarrow (\mu^{t+}_i(p^{(k+1)}), \mu^{t-}_i(p^{(k+1)})).
\end{align}
After the update, the state vector becomes
\begin{align}
\tilde{x}^{*t}_{i} = ( l^{*t+}_{i}, l^{*t-}_{i}, b^{*t+}_{i}, b^{*t-}_{i}, m^{(k+1)t+}_{i},m^{(k+1)t-}_{i}, g^{*t+}_{i}, g^{*t-}_{i} ).
\end{align}
This state vector usually violates the constraints, i.e., $\tilde{x}^{*t}_{i} \notin {\cal X}_i $.
Therefore, each agent has to reconfigure $x^t_i$ to ensure satisfaction of the constraints under the condition that $(m^{t+}_i, m^{t-}_i) = (m^{(k+1)t+}_i, m^{(k+1)t-}_i) $.
The map $\text{\bf proj}_{{\cal X}_i} (\cdot| m^{t+}_i =  m^{(k+1)t+}_i, m^{t-}_i = m^{(k+1)t-}_i )$ denotes a projection  to ${{\cal X}_i}$ satisfying
$(m^{t+}_i, m^{t-}_i) = (m^{(k+1)t+}_i, m^{(k+1)t-}_i) $ for all $t \in {\cal T}$.
In the paper, the reconfiguration is performed by solving the sub-problem with the additional constraints, i.e., $(m^{t+}_i, m^{t-}_i) = (m^{(k+1)t+}_i, m^{(k+1)t-}_i) $.

\begin{algorithm}[t]
\caption{Iterative update in LFS-DA}
\label{alg:LFS-DA}
\begin{algorithmic}
\STATE $k \leftarrow  0$
\STATE Initialize the price profile $p^{(k)}  = (p_t^{(k)})_{t \in {\cal T}}$
\REPEAT
\STATE // Each agent solves its sub-problem~\aref{eq:lfs-da-sub}, and obtains its solution.
\STATE Update $x_i^{*} = ( x^{*t}_i )_{t \in {\cal T}}=  x^{*}_{i}(p^{(k)})$.
\STATE Update $\alpha^{*t}_{i} \leftarrow \beta^{t}_{i}p^{(k)}_t + (m^{*t-}_i - m^{*t+}_i)$ .
\STATE // Each agent submits $(\alpha^{*t}_{i}, \beta^t_i)$ to the market. \\
\STATE Update $p^{(k+1)} \leftarrow \text{\bf market\_clearing}\big((\alpha^{t}_{i}, \beta^t_i)_{i \in {\cal N}, t \in {\cal T}} \big)$,
\STATE Update  $(m^{(k+1)t+}_i, m^{(k+1)t-}_i) \leftarrow  (\mu^{t+}_i(p^{(k+1)}), \mu^{t-}_i(p^{(k+1)}))  $
\STATE // Reconfiguration by each agent
\STATE Update  $x_i^{(k+1)} \leftarrow \text{\bf proj}_{{\cal X}_i} (x_i^{*} | m^{t+}_i =  m^{(k+1)t+}_i, m^{t-}_i = m^{(k+1)t-}_i)$, for each $i$
\STATE $k \leftarrow  k+1$
\UNTIL a predefined stopping criterion is satisfied. 
\RETURN 
Transact $\{x_i^{(k)} \}_{i \in {\cal N}}$  with $p^{(k)}$ as price profile
\end{algorithmic}
\end{algorithm}

The following Lemma~\ref{lemma1} indicates that the changes in the price profile obtained by $\text{\bf market\_clearing}$ can be regarded an update formula as follows.

\begin{lemma}[Change of price profile in LFS-DA]\label{lemma1}
Price profile of the LFS-DA is updated as follows.
    \begin{align}
      p_t^{(k+1)} &= p_t^{(k)} - \bar{\theta}^t_k \xi_t(p^{(k)}) \label{eq:lfb_update}\\
      \xi_t(p^{(k)})
      &:=  \sum_{i\in{\cal N}} f^{t}_{i}( x^{*t}_{i}(p^{(k)}) ) \\
      &= \Big(\sum_{i\in{\cal N}} \big(\gamma m^{*t+}_{i}(p^{(k)}) - m^{*t-}_{i}(p^{(k)}) \big) \Big)_{t\in{\cal T}},
    \end{align}
where $\bar{\theta}^t_k = \big(\beta^{t}_{I_t(p^{(k+1)}_{t})}\big)^{-1} $, and $\beta^{t}_{I_t(p_{t})} = \gamma \beta^{t}_{I_t^{+}(p_{t})} + \beta^{t}_{I_t^{-}(p_{t})}$.
\end{lemma}
\begin{proof}
For all $t \in \mathcal{T}$,
\begin{align}
     p^{(k+1)}_{t}
      &=      \frac
      { \gamma \alpha^{*t}_{I_t^{+}(p^{(k+1)}_{t})} + \alpha^{*t}_{I_t^{-}(p^{(k+1)}_{t})} }
      { \gamma \beta^{t}_{I_t^{+}(p^{(k+1)}_{t})} + \beta^{t}_{I_t^{-}(p^{(k+1)}_{t})} }\\
      &= \frac
      { \gamma\sum_{i\in I_t^{+}(p^{(k+1)}_t)} \alpha^{*t}_{i} +\sum_{i\in I_t^{-}(p^{(k+1)}_t)} \alpha^{*t}_{i} }
      { \gamma \beta^{t}_{I_t^{+}(p^{(k+1)}_{t})} + \beta^{t}_{I_t^{-}(p^{(k+1)}_{t})} }\end{align}
\begin{align}
      &= \frac
      { \gamma\sum_{i\in I_t^{+}(p^{(k+1)}_t)}  (\beta^{t}_{i}p^{(k)}_t + (m^{*t-}_i - m^{*t+}_i))  +    \sum_{i\in I_t^{-}(p^{(k+1)}_t)} ( \beta^{t}_{i}p^{(k)}_t + (m^{*t-}_i - m^{*t+}_i))    }
      { \gamma \beta^{t}_{I_t^{+}(p^{(k+1)}_{t})} + \beta^{t}_{I_t^{-}(p^{(k+1)}_{t})} }  \\
    &= p^{(k)}_t -   \big(\beta^{t}_{I_t(p^{(k+1)}_{t})}\big)^{-1} \big(\gamma\sum_{i\in{\cal N}}m^{*t+}_{i} - \sum_{i\in{\cal N}}m^{*t-}_{i} \big)\\
    &= p_t^{(k)} - \bar{\theta}^t_k \xi(p^{(k)}) .
\end{align}
\end{proof}

\begin{theorem}
This update formula \aref{eq:lfb_update} of the LFS-DA is the same as that of the dual decomposition-based RTP algorithm \aref{eq:dual_update} except for the constant factor $\bar{\theta}^t_k$.
\end{theorem}
\begin{proof}
Comparing \aref{eq:dual_update} with \aref{eq:lfb_update} confirms the statement to be true.
\end{proof}

Therefore, if we consider the determination of the price in LFS-DA as a type of price update \aref{eq:dual_update} in dual decomposition-based RTP, the price change in LFS-DA becomes equal to that of RTP where $\theta^t_k = \big(\beta^{t}_{I_t(p^{(k+1)}_{t})}\big)^{-1}$.

In summary, in LFS-DA, if the auctioneer simply settles a clearing price profile to balance demand and supply, it implicitly solves the master problem of the dual problem \aref{eq:dual_master} in the same way as in the RTP algorithm.
This means that the changes in the price profile obtained by the LFS-DA mechanism become the same as those obtained by the RTP derived from the dual decomposition framework. Note that the sub-problems in RTP and in LFS-DA are completely the same and are solved in the same way. (see  \aref{dual_sub} and \aref{eq:lfs-da-sub}).
Therefore, both of the master problems and the sub-problems are the same and are solved in the same way in the RTP and the LFS-DA except for the constant factor of the sub-gradient in the master problem.
Our main theorem therefore indicates that  LFS-DA can increase social welfare in the same way as the conventional RTP without requiring any central price controller.

\section{EXPERIMENT}\label{sec6}
	This section describes our numerical example of the LFS-DA mechanism. In addition, we compare its results with those obtained for the dual decomposition-based RTP algorithm through a simulation experiment.

\subsection{Experimental conditions}
In this experiment, we set the number of prosumers to $N=20$ and the number of time slots to $T=24$.
The utility functions are set to be
\begin{align}
D_i^t(l):=
	\left\{
	\begin{array}{ll}
		\displaystyle\omega_i^t l -\frac{\theta_i^t}{2}l^2	&\displaystyle \Bigl( 0\le l\le \frac{\omega_i^t}{\theta_i^t}\Bigr) \\
		\displaystyle\frac{(\omega_i^t)^2}{2\theta_i^t}						&\displaystyle \Bigl( \frac{\omega_i^t}{\theta_i^t}< l\Bigr),
	\end{array}
	\right. \label{demandprofile}
\end{align}
where $\omega_i^t$ and $\theta_i^t>0$ are given constants~\cite{Samadi2010,fahrioglu2001using}.
In addition, we set $\theta_i^t = 30$ and $\omega_i^t=10$ for all $t\in \mathcal{T}, i\in\mathcal{N}$.

We considered all generators to be PV systems; therefore, no variable costs were taken into account.
The generation costs $C^t_i$ are defined by $
C^t_i(l_i^{t-}):=0\ \ (\forall l_i^{t-}\in [0,	l_i^{t-,{\rm max}}])$.
We used the PV generation profiles that were measured at 20 houses in Higashi Ohmi city during Autumn 2010 as the maximum value of the per-slot energy production profile $l_i^{t-,{\rm max}}$~\cite{kishi2012}.
Fig.~\ref{fig:PV} shows $l_i^{t-,{\rm max}}$ for each prosumer $i$.
We set $s_i^{{\rm init}}=0, s_i^{{\rm max}}=5, b_i^{+,{\rm max}} = 1, b_i^{-,{\rm max}} = 1$, and $\eta_i=0.7$ for all $i\in\mathcal{N}$.
We set the efficiency of transmitting electricity within i-Rene to $\gamma=0. 8$, and the parameters of the sub-gradient method to $\theta_k = 0.1$.
We set the maximum values to $m_i^{+,{\rm max}}=5$, and $m_i^{-,{\rm max}}=5$ for all  $i\in\mathcal{N}, t\in\mathcal{T}$.
We set the prices of the outside grid to $p_t^{G-} = 20$ and $p_t^{G+} = 0$ for all $t\in\mathcal{T}$.

In this experiment, two baseline methods were compared with the LFS-DA.
The first, {\it RTP (dual composition)}, employed Algorithm~\ref{alg:dual}. However, the demand and supply were usually not balanced before the algorithm converged to
the optimal solution, in which case the utility conducting the RTP would compensate
the difference by buying the deficit electricity from or selling the surplus electricity to the outside grid, respectively. The cost for the compensation was assumed to be redistributed equally among the agents.
{\it RTP without compensation} shows the social welfare before the difference is compensated. Although dual decomposition-based RTP could increase prosumers' welfare, it was reduced because of compensation by the utility.
The second method is {\it without trading}, that is, electricity trading does not occur in the regional market. Instead, agents attempt to
increase their welfare by using their batteries and elasticity of demand. The {\it optimal} is numerically obtained by solving the primal problem
using a solver.

\begin{figure}[t]
\centering
\includegraphics[width=0.5\textwidth,keepaspectratio=true]{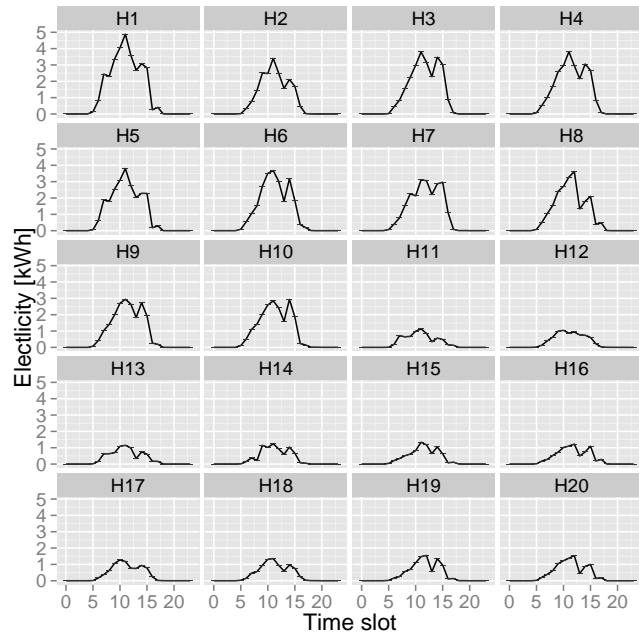}\label{fig:lt}
\caption{Amount of power generated $l_i^{t-,{\rm max}}$ by 20 prosumers' PV cells}
\label{fig:PV}
\end{figure}

\subsection{Result}

The value of the social welfare is shown in Figure~\ref{fig:social_welfare} as a function of time. This clearly shows that the LFS-DA outperforms the RTP algorithm based on dual decomposition. Especially, in the early stages of the iteration, the LFS-DA was found to outperform the RTP algorithm significantly. Although the LFS-DA mechanism was able to obtain an exact balance between demand and supply, the RTP algorithm was unable to modify prosumers' bids and asks reactively.
The cost associated with the lack of compensation had a negative impact on the overall social welfare, although each agent would still be able to increase its own welfare.

The changes in price profiles are shown in Figure~\ref{fig:price}. These figures show that the RTP algorithm and the LFS-DA obtained the same price profile except for numerical errors as our theory suggested.
Figure~\ref{fig:welfare_ratio} shows the rate at which each agent's welfare increased compared with the {\it without trading} condition. As shown, LFS-DA increased the welfare of all the agents in contrast with the RTP algorithm, which caused
several agents to reduce their welfare because of the compensation for the imbalance in demand and supply. Figure~\ref{fig:load} shows the electricity consumption of three representative agents. For the {\it without trading} condition, the electricity price for agents is fixed at $p^{G-}_t$. In contrast, RTP and LFS-DA resulted in fluctuating price profiles and leveled consumption profiles. This figure shows that LFS-DA functions in the same way as the RTP by controlling electricity consumption using price signals.

\begin{figure}[t]
\centering
\includegraphics[width=0.8\textwidth,keepaspectratio=true]{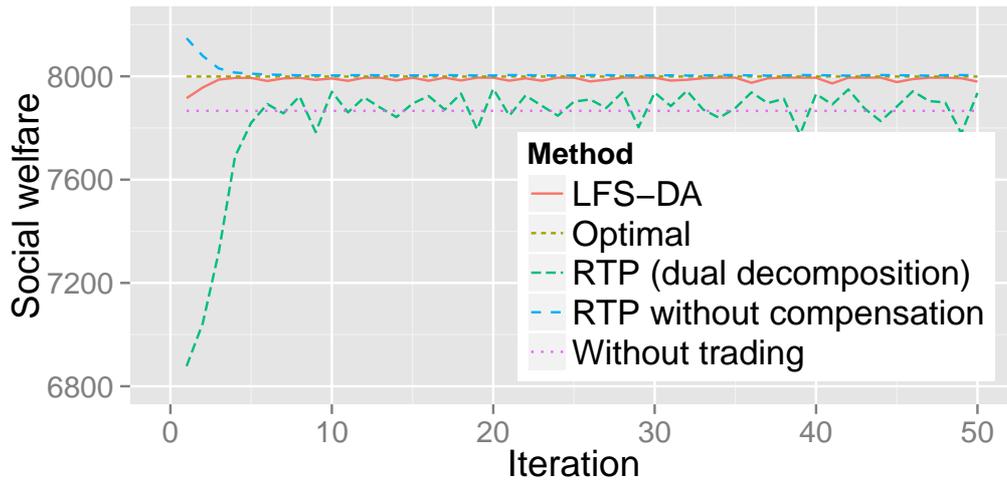}\label{fig:alpha_01}
\caption{Social welfare after iterations with $(\sum_{i \in \mathcal{N}} \beta^t_i )^{-1}= \theta _k = 0.1$}
\label{fig:social_welfare}
\end{figure}
\begin{figure}[t]
\centering
\includegraphics[width=0.5\textwidth,keepaspectratio=true]{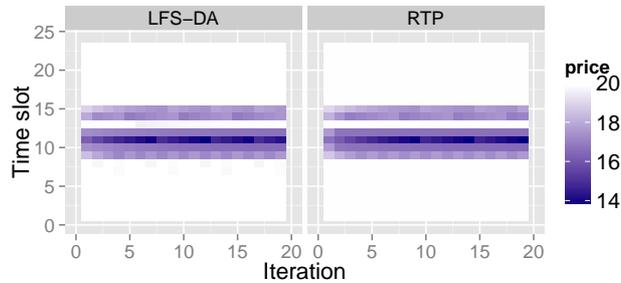}\label{fig:price}
\caption{Price profiles determined by iteration with LFS-DA and with RTP}
\label{fig:price}
\end{figure}

\begin{figure}[t]
\centering
\includegraphics[width=0.8\textwidth,keepaspectratio=true]{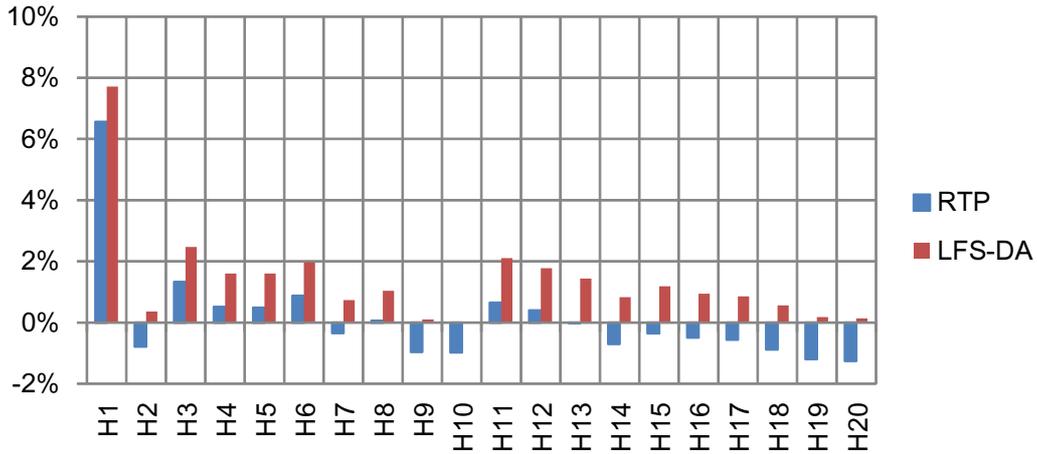}
\caption{Rate of increase of each agent's welfare compared with the {\it without trading} condition }\label{fig:welfare_ratio}
\end{figure}

\begin{figure}[t]
\centering
\includegraphics[width=0.8\textwidth,keepaspectratio=true]{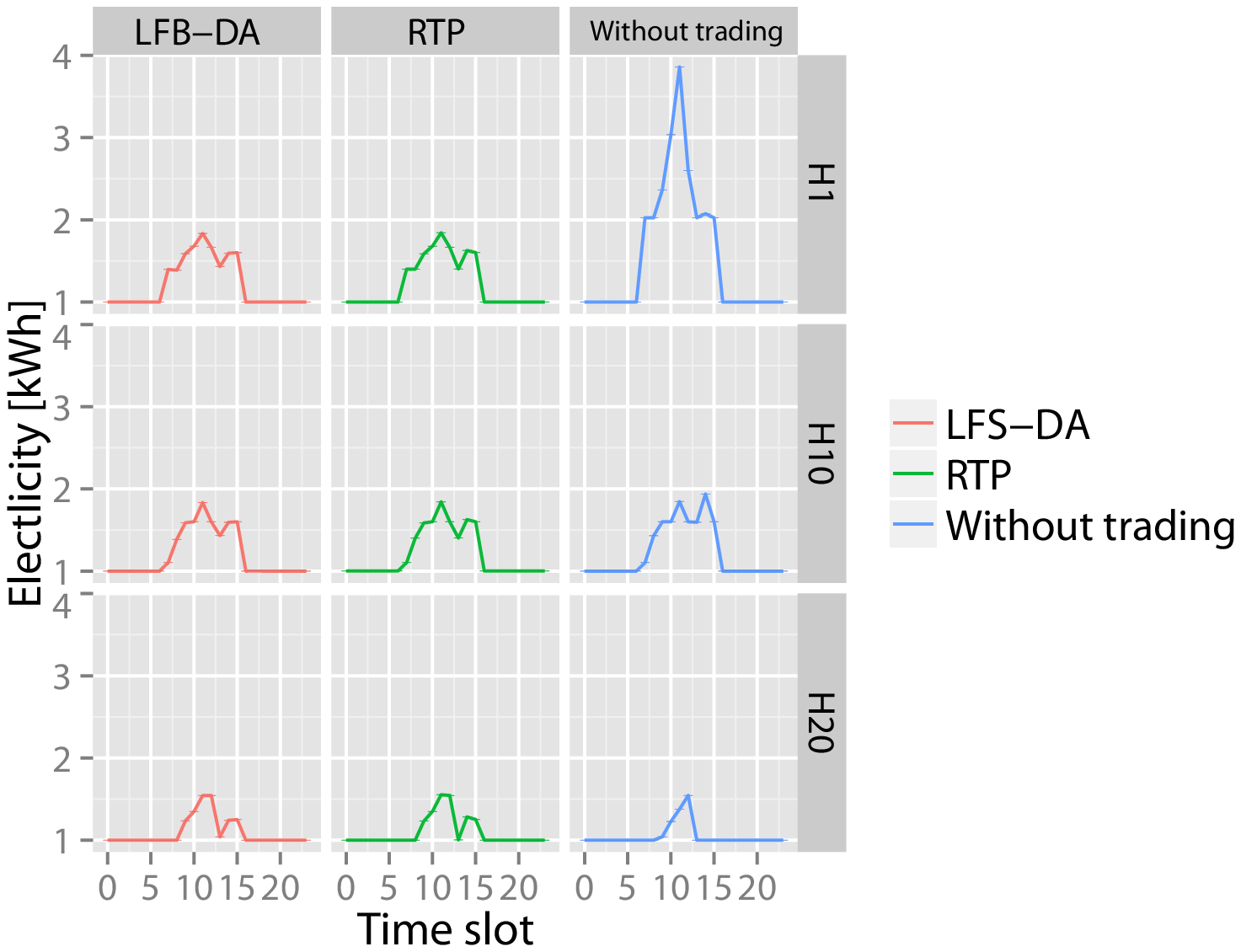}
\caption{Electricity consumption profile by three houses}\label{fig:load}
\end{figure}

\section{CONCLUSION}\label{sec7}
This paper presents a linear function submission-based double-auction (LFS-DA) algorithm for a regional prosumers' electricity network named i-Rene. 
This mechanism can achieve an exact balance between electricity demand and supply for each time slot throughout the learning phase.
It was proved that LFS-DA can solve the primal problem of maximizing the social welfare of the regional electricity network in the same way as the RTP algorithm, but without any central price setter, i.e., a utility or a large electricity company. This means that the price controller for RTP formed in a bottom-up manner as a result of the use of the double-auction mechanism can maximize the social welfare of the regional prosumers' electricity network.
This paper presents a clarification of the relationship between the RTP algorithm derived on the basis of a dual decomposition framework.

The relationship between the dual decomposition-based RTP algorithm and the simple LFS-DA algorithm  for i-Rene was clarified.
Specifically, we proved that the changes in the price profile with the LFS-DA mechanism are equal to those achieved by RTP controlled by the algorithm derived from the dual decomposition framework except for a constant factor. A simulation experiment was used to demonstrate the proposed mechanism numerically.

Our main result, i.e., Theorem 1, is very simple but powerful. We used it to show that the simple LFS-DA is sufficient for achieving day-ahead optimization and planning for i-Rene. We also showed that there is no need for a central utility to perform RTP to balance demand and supply; instead, social welfare was maximized by setting an appropriate price profile. Thus, the mechanism functions in a fully decentralized manner.

In this paper, we fixed  $\{ \beta^t_i \}_{i \in {\cal N}, t \in {\cal T}}$ which are the coefficients of the first order of the linear demand and supply functions. The parameters originally represented the price elasticity of prosumers' demand and supply.
Interestingly, our main result showed these parameters to perform the function of the learning rate in RTP. This means that the price elasticity of prosumers' demand and supply affects the convergence of the iterative calculation of the price profile in i-Rene.

Our main theorem showed the theoretical connection between research pertaining to a multi-agent-based microgrid using linear function-bidding, e.g., PowerMatcher~\cite{kok2012}, and a control theoretical approach to RTP.  The connection is expected to lead to renewed collaboration and fruitful integration between the two research areas. This formulation also provides us with the possibility to extend the LFS-DA mechanism beyond its current capabilities.
For example, Yo et al. introduced chance constraints to dual decomposition-based RTP~\cite{Yo2013} and a similar extension should be possible for the LFS-DA approach. In future, we aim to develop a double-auction mechanism, which could automatically control a regional electricity network of prosumers even in a noisy environment.

\appendix
\section*{APPENDIX}
\section{Dual Decomposition and Real-Time Pricing}\label{ap_sec1}
In this section, we derive the dual decomposition-based RTP algorithm on the basis of Problem \ref{prob:PP}, i.e., the primal problem. The complete description of the $i$-th agent's feasible set ${\cal X}_i$, except for the network constraint \aref{primal3}, becomes the following.
  \begin{align}
    {\cal X}_i &:=
    \{x_i \in \mathbb{R}^{8T}| \\
    &~~~~~~~h^{tj}(x_i) \le 0 ~~~\forall t \in {\cal T}, j\in\{1,...,16\} , \label{primal5}\\
    &~~~~~~~h^{t17}(x_i) = 0 ~~~\forall t\in{\cal T} \},\label{primal6}
       \end{align}
  \begin{align}
    h^{t1}(x_i) &:= l^{t+,\min}_{i} - l^{t+}_{i},&
    h^{t2}(x_i) &:= - l^{t-}_{i}, &
    h^{t3}(x_i) &:= - b^{t+}_{i}, \nonumber \\
    h^{t4}(x_i) &:= - b^{t-}_{i}, &
    h^{t5}(x_i) &:= - m^{t+}_{i}, &
    h^{t6}(x_i) &:= - m^{t-}_{i}, \nonumber \\
    h^{t7}(x_i) &:= - g^{t+}_{i}, &
    h^{t8}(x_i) &:= - g^{t-}_{i}, &\nonumber 
           \end{align}
  \begin{align}
    h^{t9}(x_i) &:= l^{t-}_{i} - l^{t-,\max}_{i}, &
    h^{t10}(x_i) &:= b^{t+}_{i} - b^{+,\max}_{i}, &
    h^{t11}(x_i) &:= b^{t-}_{i} - b^{-,\max}_{i},\nonumber 
                     \end{align}\begin{align}
    h^{t12}(x_i) &:= m^{t+}_{i} - m^{+,\max}_{i}, &
    h^{t13}(x_i) &:= m^{t-}_{i} - m^{-,\max}_{i}, &
    h^{t14}(x_i) &:= g^{t-}_{i} - g^{-,\max}_{i}, \nonumber 
       \end{align}
  \begin{align}
    h^{t15}(x_i) &:= -s^{\text{init}}_{i} - \sum^{t}_{k=1} (\eta_i b^{k+}_{i} - b^{k-}_{i}), &
    h^{t16}(x_i) &:= s^{\text{init}}_{i} + \sum^{t}_{k=1} (\eta_i b^{k+}_{i} - b^{k-}_{i}) - s^{\max}_{i}, \nonumber \\
    h^{t17}(x_i) &:= l^{t+}_{i} - l^{t-}_{i} + b^{t+}_{i} - b^{t-}_{i} + m^{t+}_{i} - m^{t-}_{i} + g^{t+}_{i} - g^{t-}_{i}, &
    f^{t}_{i}(x_i) &:= \gamma m^{t+}_{i} - m^{t-}_{i} .\nonumber 
  \end{align}

In the above formulas, $h^{t1} , \ldots , h^{t14}$ are constraints for the domains of the state vector $x^t_i$.
The two constraints for the battery profiles,  $h^{t15}$ and $h^{t16}$, represent the storage capacity constraints. The constraint $h^{t17}$ represents the balance of the electricity flow measured by the $i$-th agent's smart meter.

Network optimization problems with a decomposable structure, such as the primal problem, can be solved by using the dual decomposition technique~\cite{Palomar2006}.
This technique decomposes a primal problem into independent sub-problems representing agents' selfish optimization behavior, and a master problem that optimizes prices to balance the demand and supply. 
Each of the decomposed sub-problems can be solved independently. 
The solution of the primal problem and that of the dual problem are known to be the same if Slater's theorem is satisfied~\cite{bertsekas1999nonlinear}.
In addition to that, the Lagrange multipliers, which appear through dual decomposition, can be interpreted as a ``price'' in our economic system.
Within the context of the studies of smart grids and power systems, the dual decomposition technique provides the mathematical basis of RTP.
The dual problem corresponding to the primal problem becomes
\begin{problem}[Dual problem]
  \begin{align}
    \underset{\lambda\in\mathbb{R}^T}{\minimize } &\text{ } g(\lambda), \label{dual11}\\
    g(\lambda) &:= \sum_{i\in{\cal N}} \underset{x_i\in{\cal X}_i}{\sup}L_{i\lambda}(x_i), \label{dual12}\\
    L_{i\lambda}(x_i) &:= \phi_i(x_i) + \sum_{t\in{\cal T}} \lambda_t f^{t}_{i}(x_i).\label{dual13}
  \end{align}
\end{problem}
Generally, it is known that the Lagrange multiplier $\lambda := ( \lambda_t )_{ t \in {\cal T}}$ represents the ``price'' of goods, which is constrained by an equation relaxed by $\lambda_t$~\cite{Palomar2006}.
A comparison between ~\aref{dual13} and~\aref{eq:welf1} provides us with a clear understanding of this relationship.
In this problem, $\lambda_t$ represents a buyer's price of electricity traded in the local electricity market at time $t$.
The simultaneous solution of the sub-problems and the master problem results in the solution of the dual problem and this is achieved in practice by adopting an iterative optimization technique.

If the Lagrange multipliers $\lambda$ are given, the objective function~\aref{dual12} of the dual problem  can be divided into $N$ objective functions. Each objective function for each agent can be solved independently. 
\begin{problem}[Sub-problems]
  \begin{align}
    \underset{x_i\in\mathbb{R}^{8T}}{\maximize} &\text{ } L_{i\lambda}(x_i), \label{dual_sub}\\
    \subjectto
    &\text{ } x_i \in {\cal X}_i.\nonumber
  \end{align}
\end{problem}
This problem is referred to as a sub-problem of the dual problem~\aref{dual11}.
From the viewpoint of game theory or micro-economics, if each agent can be regarded as rational, i.e., selfish, each sub-problem
is expected to be solved autonomously by each agent.
If the price of electricity in the market is given as $\lambda$, one of the optimal solutions $x^*_i(\lambda)$ to the problem~\aref{dual_sub} can be obtained as follows.

  \begin{align}
    x^*_i (\lambda) = \underset{x_i \in {\cal X}_i}{\argmax} \text{ } L_{i\lambda}(x_i) \label{dual_sub*}
  \end{align}
   where $x^{*t}_{i} = ( l^{*t+}_{i}, l^{*t-}_{i}, b^{*t+}_{i}, b^{*t-}_{i}, m^{*t+}_{i}, m^{*t-}_{i}, g^{*t+}_{i}, g^{*t-}_{i} )$ and $x^{*}_{i} = (x^{*t}_{i})_{t \in {\cal T}}$ .

The dual problem is solved by optimizing the price profile $\lambda\in\mathbb{R}^T$.
For RTP formulation purposes, we assume there is a central utility determining the price profile in the regional electricity market.
On the basis of each agent's optimal strategy $(x^*_i(\lambda))_{i\in{\cal N}}$ 
, the central utility has to solve the following problem.
\begin{problem}[Master problem]
   \begin{align}
      \underset{\lambda\in\mathbb{R}^T}{\minimize} &\text{ } g(\lambda),\label{eq:dual_master}\\
      g(\lambda) &= \sum_{i\in{\cal N}} L_{i\lambda}(x^*_i(\lambda)).
    \end{align}
\end{problem}
This problem is referred to as the {\it master problem} and its solution corresponds to searching for an adequate price that would balance demand and supply.


One of the most common procedures for solving the dual problem is as shown in Algorithm~\ref{alg:dual}.
The Lagrange multiplier $\lambda_t^{(k)}$ is updated by using a known sub-gradient method~\cite{Palomar2006,Samadi2010}
If we adopt this sub-gradient method, $\lambda_t^{(k)}$ is updated as follows in the dual problem:
    \begin{align}
      \lambda_t^{(k+1)} &= \lambda_t^{(k)} - \theta_k \xi(\lambda^{(k)}) \label{eq:dual_update_lambda}\\
      \xi(\lambda^{(k)})
      &:= \Big( \sum_{i\in{\cal N}} f^{t}_{i}( x^{*}_{i}(\lambda^{(k)}) ) \Big)_{t\in{\cal T}}\\
      &= \Big(\sum_{i\in{\cal N}} \big(\gamma m^{*t+}_{i}(\lambda^{(k)}) - m^{*t-}_{i}(\lambda^{(k)}) \big) \Big)_{t\in{\cal T}},
    \end{align}
where $\theta_k > 0$ is the learning rate of the sub-gradient method.
By updating $x_i$ and $\lambda$ iteratively, the numerical solution of the dual problem can be obtained. The solution of the dual problem is also a solution of the primal problem. Therefore, the social welfare of the network is expected to be maximized.




\authorcontributions{Author Contributions}
Tadahiro Taniguchi proved the main theorem and wrote the paper; Koki Kawasaki performed the experiments; Yoshiro Fukui derived basic formulation; Tomohiro Takata and Shiro Yano contributed analysis. 

\conflictofinterests{Conflicts of Interest}
The authors declare no conflict of interest.

\bibliographystyle{mdpi}
\makeatletter
\renewcommand\@biblabel[1]{#1. }
\makeatother
\bibliography{smgr}
%
%
%



%


%

\end{document}